\documentclass[a4paper,final]{cas-dc}
\usepackage[numbers, sort]{natbib}

\usepackage{xcolor}
\usepackage{algorithm}
\usepackage{algorithmic}
\usepackage{enumitem}
\usepackage{bm}

\newproof{proof}{Proof}
\newtheorem{lemma}{Lemma}
\newtheorem{theorem}{Theorem}
\newtheorem{solution}{Solution}
\newtheorem{corollary}{Corollary}
\newtheorem{definition}{Definition}[section]

\def\tsc#1{\csdef{#1}{\textsc{\lowercase{#1}}\xspace}}
\tsc{WGM}
\tsc{QE}
\tsc{EP}
\tsc{PMS}
\tsc{BEC}
\tsc{DE}

\begin{document}
\let\WriteBookmarks\relax
\def\floatpagepagefraction{1}
\def\textpagefraction{.001}
\shorttitle{}

\title [mode = title]{Fast Collaborative Inference via Distributed Speculative Decoding}                      



\author[1]{Ce ZHENG}[orcid=0000-0001-7389-6010]
\ead{zhengc@pcl.ac.cn}

\author[2]{Ke ZHANG}
\cormark[1]
\ead{zhangke@moegi.waseda.jp}

\author[3]{Chen SUN}
\cormark[2]
\ead{chen.sun@sony.com}

\author[3]{Wenqi ZHANG}
\ead{wenqi.zhang@sony.com}

\author[4]{Qiong LIU}
\ead{qiong.liu@lip6.fr}

\author[5]{Angesom Ataklity Tesfay}
\ead{angesom-ataklity.tesfay@univ-eiffel.fr}

\affiliation[1]
{
organization={Department of Network Intelligence, Pengcheng Laboratory},
city={Shenzhen},
state={Guangdong},
country={China}
}

\affiliation[2]
{
organization={Department of Computer Science and Communications Engineering, Waseda University},
city={Tokyo},
country={Japan}
}

\affiliation[3]
{
organization={Wireless Netowork Research Department, SONY Research and Development Center},
city={Beijing},
country={China}
}

\affiliation[4]
{
organization={Sorbonne Université, CNRS, LIP6},
city={Paris},
country={France}
}

\affiliation[5]
{
organization={Laboratoire Électronique Ondes et Signaux, Université Gustave Eiffel},
city={Lille},
country={France}
}

\cortext[cor1]{Corresponding author}
\cortext[cor2]{Principal corresponding author}



\begin{abstract}
Speculative decoding accelerates large language model (LLM) inference by allowing a small draft model to predict multiple future tokens for verification by a larger target model. In AI-native radio access networks (AI-RAN), this enables device--edge collaborative inference but introduces significant uplink overhead, as existing distributed speculative decoding schemes transmit full vocabulary logits at every step.
We propose a sparsify-then-sample strategy, Truncated Sparse Logits Transmission (TSLT), which transmits only the logits and indices of a truncated candidate set. We provide theoretical guarantees showing that the acceptance rate is preserved under TSLT. TSLT is further extended to multi-candidate case, where multiple draft candidates per step increase acceptance probability.
Experiments show that TSLT significantly reduces uplink communication while maintaining end-to-end inference latency and model quality, demonstrating its effectiveness for scalable, communication-efficient distributed LLM inference in future AI-RAN systems.
\end{abstract}




\maketitle

\section{Introduction}
Large language models (LLMs) have reshaped artificial intelligence, powering applications from content generation to intelligent decision support. 
Meanwhile, AI-native radio access networks (AI-RANs)—which integrate communication and computation within the wireless infrastructure—provide a promising platform for real-time, LLM-enabled intelligence~\cite{kundu2025}.
Nevertheless, realizing efficient LLM inference in such heterogeneous AI-RAN environments remains challenging: devices are constrained by limited memory, energy, and computational capabilities, whereas server-side computation suffers from latency, performance jitter, and mobility-induced disruptions. Meanwhile, the autoregressive nature of large models introduces token-wise dependencies that hinder efficient partitioning across devices and servers, including pipeline and sequence parallelism~\cite{gao2025towards}.

These limitations motivate a new \textbf{\textit{collaborative inference}} paradigm in which lightweight small language models (SLMs) on devices interact with more capable LLMs hosted at base stations or edge servers~\cite{ding2024hybrid, hao2024hybrid}. For example, \cite{ding2024hybrid} introduces an adaptive routing mechanism that allocates tasks between the SLM and LLM based on computational demand and output fidelity. Similarly, \cite{hao2024hybrid} proposes a draft-and-verify scheme where SLM token outputs are filtered through a tunable probability threshold, balancing inference cost and model accuracy. 
While earlier collaborative schemes improved responsiveness, they often did so at the cost of degraded inference accuracy. Consequently, \textbf{\textit{speculative decoding (SD)}} was proposed to accelerate inference by having the SLM generate a sequence of draft tokens autoregressively, which are then verified in parallel by the LLM on the same node, preserving output fidelity~\cite{leviathan2023fast, chen2023accelerating}.  Extending this paradigm, \textbf{\textit{distributed speculative decoding (DSD)}} accelerates collaborative inference between SLM and LLM by executing the SLM locally and offloading verification to an edge server~\cite{zhao2024edge, oh2024uncertainty, ning2025dssd, zheng2025communication}. However, DSD entails frequent communication between the SLM and LLM, including transmitting draft token sequences and their probability distributions, which can introduce substantial delays—especially under constrained uplink conditions—and may degrade overall decoding performance.

DSD was first introduced by~\cite{zhao2024edge}, demonstrating how tuning the number of draft tokens can balance communication overhead and computation, thereby reducing end-to-end latency and energy consumption in wireless scenarios. Building on this idea, \cite{shao2025ai} conducted a case study on the Vehicles-OpenImage dataset, showing that DSD using a small on-device InternVL2-2B model and a large edge InternVL2-26B model can significantly improve inference speed while maintaining output accuracy. In parallel, \cite{xie2025novel} explores a complementary collaboration paradigm, proposing HAT, a hybrid U-shaped inference and SD framework that partitions the LLM into three submodels and exchanges hidden states rather than raw tokens. Despite the advantages of DSD, its practical deployment is constrained by substantial communication costs. To tackle this bottleneck, \cite{oh2024uncertainty} proposes an \textit{uncertainty-aware opportunistic hybrid language model (U-HLM)} that skips uplink transmission for tokens likely to be accepted, while \cite{oh2025communication} further extends this work with a Top-$K$ truncated vocabulary compression scheme; however, both methods introduce extra computation and inevitably degrade inference accuracy. Such degradation is particularly problematic for robotic control, where even small errors can propagate through the perception–decision–control pipeline, destabilizing action generation, causing trajectory drift, and potentially leading to unsafe behaviors~\cite{park2025action, park2025uncertainty}. 
Such concerns have spurred growing interest in \textit{lossless} DSD, with the aim of reducing communication overhead while preserving exact LLM outputs, though research in this direction remains relatively scarce. For instance, \cite{ning2025dssd} introduces \textit{distributed split speculative decoding} (DSSD), which partitions the verification phase between device and edge, sending a single vocabulary distribution in a downlink instead of multiple distributions in an uplink as in DSD. 
What is more, existing DSD-based approaches, ranging from lossy methods such as U-HLM to the lossless DSSD, have been limited to single-candidate decoding and have not considered multi-candidate approaches, in which multiple draft token sequences are generated and verified in parallel to improve decoding efficiency and reduce end-to-end latency~\cite{yang2024multi, miao2024specinfer, khisti2024multi, lu2024improving}. 

To mitigate the excessive vocabulary transmission overhead in DSD, we propose a \textit{sparsify-then-sample} strategy that performs Top-$K$ or Top-$p$ sampling over the vocabulary. Specifically, transmitting the full probability distribution scales linearly with the vocabulary size; for instance, a 32K vocabulary with FP16 representation generates roughly 500 KB per token. By first truncating the distribution to the Top-$K$ or Top-$p$ entries and then re-normalizing, our method significantly reduces the effective vocabulary size and communication cost while ensuring lossless inference. Unlike the prior work \cite{oh2025communication}, which directly truncates the distribution by retaining only the Top-$K$ entries—thereby distorting the residual probabilities and resulting in lossy inference—our \textit{sparsify-then-sample} scheme preserves the full probability mass of the truncated distribution, guaranteeing lossless inference.
Building on this design, we extend the approach to the \textit{multi-candidate} (MC) regime, enabling parallel generation and verification of multiple draft sequences. A detailed algorithmic workflow is provided; although the procedure is complex, it is clearly structured and readily understandable. This is accompanied by a comprehensive analysis, with particular emphasis on its impact on the acceptance rate. In the MC regime, the communication overhead scales with the size of the generated token tree, depending on both its depth and breadth.~
In the experimental evaluation, we first computed the cumulative probability mass of the Top-$k$ entries under different values of $k$, confirming the sparsity of the output distribution and thereby demonstrating the feasibility of our sparsification approach. Next, we reported the acceptance rates corresponding to various $k$ values, which empirically validate our theoretical analysis. Finally, we measured the speedup achieved by our method, showing the performance improvement compared to inference using a single large model.

Therefore, our contributions are: \textbf{\textit{1)}. Sparsify-then-sample for lossless inference}: a sparsification strategy that truncates and re-normalizes the output distribution, significantly reducing communication overhead while ensuring lossless inference with analysis on robustness of acceptance rate.
\textbf{\textit{2)}. Extension to MC regime}: Generalization to multiple draft sequences, with a clearly structured yet comprehensive algorithm that maintains robust acceptance rates.
\textbf{\textit{3)}. Empirical validation}: Experiments confirm output sparsity, validate acceptance-rate robustness, and demonstrate speedup compared to standalone LLM inference.

The remainder of this paper is organized as follows. Section~\ref{sec:System model} introduces the system model and problem formulation. Section~\ref{sec:Truncated TSLT} presents the our Sparsify-then-sample Strategy. Section~\ref{sec:MU-DSD} extends the approach to multi-candidate speculative decoding scenarios. Section~\ref{sec:experiment} details the experimental setup and discusses the results. Finally, Section~\ref{sec:conclusion} concludes the paper and outlines potential directions for future work.

\section{System model}\label{sec:System model}
We first consider a single-candidate DSD (SC-DSD) scenario, where a SLM runs locally on the device, while a more powerful LLM resides at the base station (BS)~\cite{oh2024uncertainty, zhao2024edge}. Both models share the same tokenizer and hence the same vocabulary $\mathcal{V}$. Fig.~\ref{fig:DSD} illustrates the workflow of the DSD process. For simplicity and in line with most prior work, we assume that both models share the same tokenizer and hence the same vocabulary $\mathcal{V}$. 

\noindent\textbf{Remark.} This shared-vocabulary assumption is strong: in practice, the device and server models may have different vocabularies, which could lead to token mapping issues or decoding mismatches. Recent studies have proposed strategies to handle heterogeneous vocabularies in speculative decoding, such as string-level remapping and adaptive drafter mechanisms~\cite{timor2025accelerating, ramakrishnan2025omnidraft}. Implementing such cross-vocabulary support is left for future work, while our current study focuses on the simpler shared-vocabulary setting for clarity and tractability.

\subsection{Distributed speculative decoding}
\label{sec:V-DSD}
\begin{figure*}[htbp]
\begin{center}
\centerline{\includegraphics[width=0.75\textwidth]{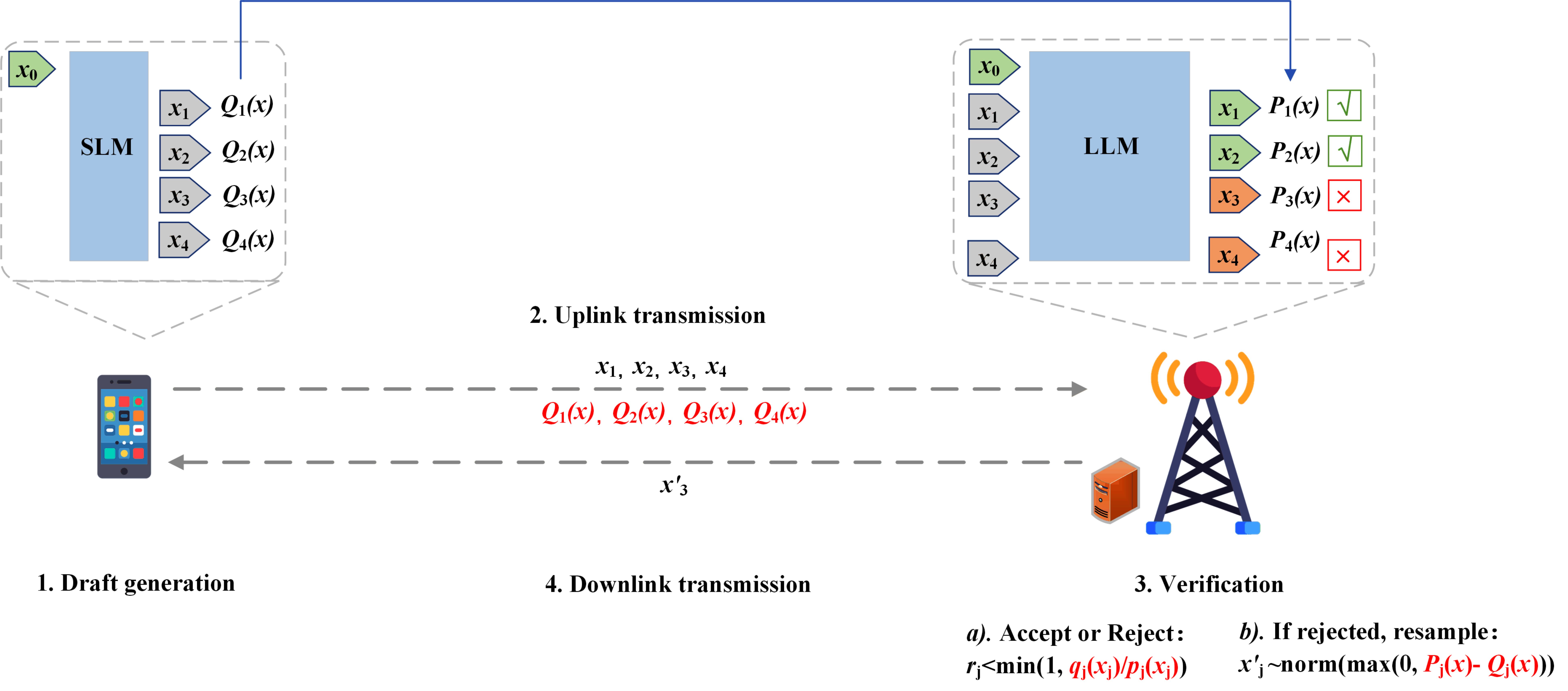}}
\caption{Single Candidate Distributed Speculative Decoding}
\label{fig:DSD}
\end{center}
\end{figure*}

As shown in Fig.~\ref{fig:DSD}, the SLM produces $L$ candidate tokens in an autoregressive manner, while the LLM performs parallel verification, retaining the accepted tokens and regenerating replacements for those that are rejected. The detailed procedure is presented in Algorithm~\ref{alg:DSD}, \ref{alg:TokSeqDraft}, \ref{alg:TokSeqVeri}. Specifically, \textit{1).} SLM first generates $L$ tokens $x_i \sim Q_(x)$ as described in Algorithm~\ref{alg:TokSeqDraft}, and then \textit{2).} sends $\mathcal{X} = \left[ x_1,\cdots,x_{L} \right]$ and $\mathcal{Q} = \left[Q_1(x), \cdots, Q_{L}(x) \right]$ to BS.\footnote{To reduce communication overhead, the device transmits only the indices of the drafted tokens within the vocabulary rather than their full string representations. For clarity of presentation, we still depict tokens in Fig.~\ref{fig:DSD} and Alg.~\ref{alg:DSD}.} \textit{3).} At the BS, following Algorithm~\ref{alg:TokSeqVeri}, the LLM is first run once to get $\mathcal{P} = \left[ P_1(x), \cdots, P_{L+1}(x) \right]$. Then it verifies the received tokens $\mathcal{X}$ in parallel and outputs the verified token $\mathcal{X}_{veri} = \left[ x_1,\cdots,x_j \right]$. \textit{4).} The BS sends back $x_j$ together with its index $j$ for prefix update at the device.\footnote{As $\mathcal{X}$ is already stored locally, only $x_j$ and $j$ are transmitted.} The procedure continues until completion.

\begin{algorithm}[ht]
\caption{Single-Candidate DSD}
\label{alg:DSD}
\begin{algorithmic}
\STATE \hspace{-1.2em} {\bf $\triangleright$ Device (Token Sequence Drafting):}
\STATE \hspace{-1.2em} {\bfseries Input:} Prefix, SLM -- $M_q$, token length -- $L$
\STATE  $\mathcal{X},\mathcal{Q} \gets \textsc{\bf TokSeqDraft}(Prefix, L)$  
\STATE \hspace{-1.2em} {\bfseries Output:} $\mathcal{X}=\left[ x_1,\cdots,x_
{L} \right]$ and $\mathcal{Q}=\left[ Q_1(x), \cdots, Q_{L}(x) \right]$
\STATE \vspace{-2mm} \hspace{-1em}\hrulefill
\STATE \hspace{-1.2em} {\bf $\triangleright$ Uplink Transmission:}
    \STATE Send tokens $\mathcal{X}$ and $\mathcal{Q}$ to the BS via uplink transmission
\STATE \vspace{-2mm} \hspace{-1em}\hrulefill
\STATE \hspace{-1.2em} {\bf $\triangleright$ Edge (Token Sequence Verification): }
\STATE {\bfseries Input:} Prefix, LLM -- $M_p$, $\mathcal{X}$, $\mathcal{Q}$
\STATE $\mathcal{X}_{veri},\mathcal{Q} \gets \textsc{\bf TokSeqVeri}(Prefix, \mathcal{X}, \mathcal{Q})$
\STATE {\bfseries Output:} $\mathcal{X}_{veri} = \left[ x_1,\cdots,x_j \right]$
\STATE \vspace{-2mm} \hspace{-1em}\hrulefill
\STATE \hspace{-1.2em} {\bf $\triangleright$ Downlink Transmission \& Prefix Update}
    \STATE Send $x_j$ and $j$ back to device 
    \STATE Update $\mathrm{Prefix} = \mathrm{Prefix} \oplus \mathcal{X}_{veri}$ for at both device and edge
\end{algorithmic}
\end{algorithm}

\begin{algorithm}[ht]
\caption{\textsc{TokSeqDraft}}
\label{alg:TokSeqDraft}
\begin{algorithmic}
    \STATE \hspace{-1.2em} {\bfseries Device:} Initialize $\mathbf{y}=[]$
    \STATE {\bfseries Input:} Prefix, SLM -- $M_q$, token length -- $L$
    \STATE {\bf $\triangleright$ Sample $L$ tokens autoregressively:}
    \FOR{$i=1$ \TO $L$}
    \STATE $Q_i(x)\leftarrow M_q(prefix \oplus \mathbf{y})$
    \STATE Sample $x_i \sim Q_i(x)$
    \STATE $\mathbf{y} = \mathbf{y} \oplus x_i$
    \ENDFOR
    \STATE {\bfseries Output:} $\mathcal{X} = \left[ x_1,\cdots,x_{L} \right]$ and $\mathcal{Q} = \left[ Q_1(x), \cdots, Q_{L}(x) \right]$
\end{algorithmic}   
\end{algorithm}

\begin{algorithm}[ht]
\caption{\textsc{TokSeqVeri}}
\label{alg:TokSeqVeri}
\begin{algorithmic}
\STATE \hspace{-1.2em} {\bf $\triangleright$ Verification Process}
\STATE \hspace{-1.2em} {\bfseries Edge:} Initialize $j=1$, $\mathrm{Flag} = 1$
    \STATE {\bfseries Input:} Prefix, LLM -- $M_p$, $\mathcal{X}$, $\mathcal{Q}$
    \STATE {\bf $\triangleright$ Run $M_p$ in parallel}
    \STATE $\mathcal{P}=[P_1(x), \cdots, P_{L+1}(x)] \gets M_q(prefix \oplus \mathcal{X})$
    \WHILE{$j \leq L$ \& $\mathrm{Flag}=1$}
        \STATE Sample $r_j \sim U[0,1]$
        \STATE {\bf $\triangleright$ Accept}
        \IF{$r_j< \min \left\{1, \frac{q_j(x_j)}{p_j(x_j)} \right\}$}
        \STATE $x_j$ is accepted
        \STATE $j=j+1$
        \STATE \hspace{-1.2em} {\bf $\triangleright$ Reject and resample}
        \ELSE
        \STATE $x'_j \sim \mathrm{norm} \left(\max \left(0, P_j(x)-Q_j(x)\right)\right)$
        \STATE $x_j = x'_j$
        \STATE $\mathrm{Flag}=0$
        \ENDIF
    \IF{$j=L+1$}
        \STATE Sample $x_{L+1} \sim P_{L+1}(x)$
    \ENDIF
    \ENDWHILE
    \STATE {\bfseries Output:} 
\end{algorithmic}   
\end{algorithm}

\subsection{Wireless communication}
We ignore the downlink transmission latency in our system model, as it is significantly smaller than the uplink latency and does not affect the overall delay analysis.\footnote{The downlink generally achieves higher data rates due to higher transmit power and more favorable propagation conditions. Moreover, in DSD, the downlink only delivers lightweight feedback (e.g., acceptance signals or corrected tokens).} For simplicity, we also assume that the uplink transmission of drafted token indices contributes negligibly to the total communication cost. 
 More importantly, we model the uplink using an \textit{effective data rate} \(R_{\text{up}}\), which can be interpreted as the post-HARQ goodput. This abstraction implicitly accounts for practical wireless link behaviors such as retransmissions, packet losses, and coding overhead, and is standard in edge inference and communication-efficient ML literature. Since these factors mainly scale the communication time by a constant factor, omitting an explicit link-layer model does not affect the relative speedup analysis.

Therefore, the communication time is
\begin{equation}
    T_{\text{comm}} = \frac{D_{\text{up}}} {R_{\text{up}}},
\end{equation}
where
\begin{align}
    \label{eq:Vocabulary}
    D_{\text{up}} &= L \cdot D_{\mathcal{V}}, \\ D_{\mathcal{V}} &= |\mathcal{V}| \cdot \left( b_{\text{prob}}+b_{\text{idx}} \right)
\end{align}
with $D_{\mathcal{V}}$ denoting the amount of data for a single vocabulary distribution, $|\cdot|$ the cardinality, $b_{\text{prob}}$ the bit-width of each probability value (e.g., $b_{\text{prob}} = 32$ for full precision or 16 for half precision), and $b_{\text{idx}}$ the bit-width for each token index. Using binary encoding for indices gives $b_{\text{idx}} = \left\lceil \log_2 |\mathcal{V}| \right\rceil$.

Hence, we have
\begin{equation}
T_{\text{comm}} = L \cdot |\mathcal{V}| \cdot \frac{ b_{\text{prob}} + b_{\text{idx}} }{R_{\text{up}}}.
\end{equation}

\subsection{Token throughput}
Inference latency consists of three components: the on-device SLM drafting time, the edge-side LLM verification time, and the device–edge communication time. Following~\cite{yin2024theoretical}, we refer to one execution of Algorithm~\ref{alg:DSD} as an ``\textit{oracle call},'' whose latency is given by
\begin{align}
\label{eq:T_inf}
    T_{\text{oracle}} &= L \cdot T_{\text{SLM}} + T_{\text{comm}} + T_{\text{LLM}} 
\end{align}
where $T_{\text{SLM}}$ and $T_{\text{LLM}}$ denote the runtime of a single execution of the SLM and LLM, respectively.

We follow~\cite{leviathan2023fast} to define the acceptance rate as follows.

\begin{definition}[Acceptance Rate~\cite{leviathan2023fast}]
Given a prefix, the acceptance rate $\beta_i$ of a token
$x \sim Q_i$ under speculative decoding is defined as
\begin{equation}
\beta_i = \Pr[x \text{ is accepted} \mid x \sim Q_i].
\end{equation}
The \textit{expected acceptance rate} is
\begin{equation}
\alpha = \mathbb{E}[\beta_i].
\end{equation}
\end{definition}

Under this definition, the expected number of tokens generated in one oracle call is~\cite{leviathan2023fast}
\begin{equation}
N_{\text{oracle}} = \frac{1-\alpha^{L+1}}{1-\alpha}.
\end{equation}

Therefore, the token througput is
\begin{align}
\label{eq:throughput}
\Theta_{\text{DSD}} 
&= \frac{N_{\text{oracle}}}{T_{\text{oracle}}} \notag\\
&= \frac{1-\alpha^{L+1}}{(1-\alpha)L\left(T_{\text{SLM}}+|\mathcal{V}| \frac{b_{\text{prob}} + b_{\text{idx}}}{R_{\text{up}}} \right) + T_{\text{LLM}}}
\end{align}

We introduce the speedup ratio $S_{\text{inf}}$ to quantify the performance gain of DSD over a standalone LLM:
\begin{definition}[Speedup Ratio]
\begin{equation}
\label{eq:speedup}
    S_{\mathrm{inf}} = \frac{\Theta_{\mathrm{DSD}}}{\Theta_{\mathrm{LLM}}},
\end{equation}
where $\Theta_{\mathrm{LLM}}=\frac{1}{T_{\mathrm{LLM}}}$ is the token throughput of standalone LLM.
\end{definition}

This metric enables clear comparison with the baseline, eliminates hardware-dependent factors, and is widely used in prior works.
It also facilitates bottleneck analysis and helps reveal how key parameters (e.g., $\alpha$ and $|\mathcal{V}|$) affect the achievable acceleration.

\subsection{Problem formulation}
Our objective is to maximize the token throughput achieved by DSD. Equivalently, we aim to maximize the speedup ratio \(S_{\mathrm{inf}}\), defined in \eqref{eq:speedup}. From \eqref{eq:throughput} and \eqref{eq:speedup}, the speedup can be expressed as a function of key system parameters:
\begin{equation}
    S_{\text{inf}} = S_{\text{inf}}\big( |\mathcal{V}|, \alpha, L, b_{\text{prob}}, b_{\text{idx}}, T_{\text{SLM}}, T_{\text{LLM}}, R_{\text{up}} \big).
\end{equation}

Among these parameters, \(T_{\mathrm{SLM}}\) and \(T_{\mathrm{LLM}}\) are determined by the model architectures and hardware platform. The index bit-width \(\lceil \log_2 |\mathcal{V}| \rceil\) is fixed once the vocabulary size is chosen. The probability quantization bit-width \(b_{\mathrm{prob}}\) and the draft length \(L\) are treated as predefined hyperparameters. The acceptance rate \(\alpha\), in turn, depends on the alignment between the SLM and LLM outputs, and is therefore also influenced by the effective vocabulary over which the logits are normalized.

While the hardware- and model-dependent parameters are fixed, the system performance can still be substantially improved by reducing the communication load per draft step. Notably, the vocabulary size \(|\mathcal{V}|\) jointly determines the payload of each transmitted logits vector and the acceptance behavior in speculative decoding. This makes compressing the effective vocabulary—without compromising correctness—a key opportunity for boosting throughput.

Motivated by this observation, the next section introduces a sparsify-then-sample scheme that reduces the effective vocabulary size, lowers communication overhead, and thereby improves the token throughput of DSD.

\section{Sparsify-then-sample strategy}
\label{sec:Truncated TSLT}

As shown in \eqref{eq:Vocabulary}, token generation involves uploading a distribution whose payload size is proportional to the dimensionality of the vocabulary space. Such excessive uplink transmission leads to prohibitive communication latency, thereby reducing overall inference throughput. To address this, we propose a sparsify-then-sample approach: the logits are first sparsified using truncated sampling to obtain a reduced vocabulary distribution, which is then sampled by the SLM. The resulting sparsified logits are transmitted to the LLM for verification. Based on this design, we term our method the \textbf{Truncated Sparse Logits Transmission (TSLT)} scheme, which reduces the transmitted payload while preserving the correctness of generated tokens.

Unlike the naive Top-$K$ truncation method used in \cite{oh2025communication}, which may degrade inference accuracy, our method ensures lossless inference. For completeness, we first review the truncated sampling strategies commonly employed in language generation.

\subsection{Truncated sampling}
In truncated sampling, the candidate space is restricted to a subset of the vocabulary before probabilistic sampling is performed.
Formally, given a vocabulary $\mathcal{V}$ and a probability distribution $Q(x)$ over $\mathcal{V}$, truncated sampling constructs a subset $\mathcal{V}' \subseteq \mathcal{V}$ and renormalizes the probabilities:
\begin{equation}
\label{eq:truncation}
    \hat{q}(x) = \begin{cases}
        0,  & x \in \mathcal{V'} \\
        \frac{q(x)}{\sum_{v \in \mathcal{V'}} q(v)}, & x \notin \mathcal{V'}
    \end{cases}
\end{equation}
Two widely used truncation schemes are {\bfseries Top-$K$ sampling} and {\bfseries Top-$p$ (nucleus) sampling}. Let $\left\{x^{[1]}>x^{[2]}>\dots, x^{[|\mathcal{V}|]} \right\}$ be the tokens sorted in descending order of probability $P(x)$. We define

\begin{itemize}[left=0pt]
\item {\bfseries Top-$K$ sampling}:
Select the $K$ tokens with the highest probabilities:
\begin{equation}
    \mathcal{V}_K = \left\{ x^{[1]}, x^{[2]}, \dots, x^{[K]} \right\},
\end{equation}
\item {\bfseries Top-$\rho$ (nucleus) sampling}\footnote{Also commonly denoted as "Top-$p$" sampling in the literature; here we use $\rho$ to avoid confusion with probability distributions $p(\cdot)$}:
Select the smallest set of tokens $\mathcal{V}_{\rho}$ such that their cumulative probability exceeds a threshold $p \in (0,1]$:
\begin{equation}
    \mathcal{V}_{\rho} = \left\{ x^{[1]}, x^{[2]}, \dots, x^{[n]} \right\},
\end{equation}
where $\sum_{i=1}^{n} q\left(x^{[i]}\right) < \rho$ and $\sum_{i=1}^{n+1} q\left(x^{[i]}\right) \geq \rho$.
\end{itemize}

By renormalizing over $\mathcal{V}_K$ or $\mathcal{V}_{\rho}$, truncated sampling eliminates extremely low-probability tokens while retaining sufficient diversity, thus achieving a balance between fluency and variability in generation.

\subsection{Truncated sparse logits transmission}
Leveraging truncated sampling, we propose the Truncated Sparse Logits Transmission (TSLT) scheme to reduce uplink overhead. The underlying intuition is as follows: provided that a draft token undergoes the verification process, SCa-DSD guarantees consistency with a standalone LLM. Specifically, the resulting token probability distribution is theoretically equivalent to that of the original LLM, {\bfseries irrespective of the specific form or variation of $\bm{Q(x)}$}, the output distribution from the SLM. Consequently, {\bfseries if $\bm{Q(x)}$ is sufficiently sparse, the communication cost can be significantly reduced without affecting the probabilistic equivalence with the original LLM}:

\begin{solution}[Truncated Sparse Logits Transmission]
The SLM first performs truncated sampling directly on the logits, e.g., Top-$K$ or Top-$\rho$ (nucleus) sampling. Only the selected logits and their corresponding token IDs are transmitted to the LLM for verification. 
\end{solution}

In this scheme, only the selected logits are retained, resulting in a sparse representation. By transmitting this sparse set rather than the full logits, the communication overhead is significantly reduced, while the most probable candidate tokens are preserved for verification by the LLM.
Compared with applying truncation on probabilities, where softmax normalization tends to suppress small-probability tokens and amplify rounding errors, logits-based truncation better preserves numerical precision and provides a more stable and faithful representation of the SLM's output.

In contrast to prior approaches such as \cite{oh2025communication}, which perform sampling and then apply Top-$K$ truncation on the transmitted distribution, this strategy can lead to lossy inference because the truncated set no longer reflects the full token distribution of the standalone LLM. Such methods introduce  bias. Our logits-level truncation avoids this problem by retaining the most informative logits before normalization, thereby offering a more coherent trade-off between communication efficiency and fidelity to the original LLM. However, the acceptance rate is modified. 

In contrast to prior approaches such as \cite{oh2025communication}, which perform sampling and then apply Top-$K$ truncation on the transmitted distribution, this strategy can lead to lossy inference because the truncated set no longer reflects the full token distribution of the standalone LLM. Such methods introduce distributional bias, as the retained probabilities are distorted by the sampling-and-truncation pipeline. Our logits-level truncation avoids this problem by retaining the most informative logits before normalization, thereby offering a more coherent trade-off between communication efficiency and fidelity to the original LLM. However, the acceptance rate is modified:

\begin{lemma}
[Acceptance Rate and TV Distance~\cite{leviathan2023fast}]
For the $i$-th decoding iteration, the acceptance rate satisfies
\begin{equation}
    \beta_j = \sum_{x\in\mathcal{V}} \min \left\{ Q_i(x), P_i(x) \right\}
            = 1 - TV \left( Q_i, P_i \right),
\end{equation}
where $Q_i$ and $P_i$ denote the SLM and LLM output distributions, respectively. And
\begin{equation}
    TV\left(Q_i,P_i\right)
    = \frac{1}{2} \sum_{x\in\mathcal{V}} \left|Q_i(x)-P_i(x)\right|
\end{equation}
is the total variation distance between the two distributions.
\end{lemma}

When truncated sampling method is employed, the acceptance rate becomes
\begin{equation}
\hat{\beta_i} = \sum_x \min \left\{ \hat{q}_i(x), p_i(x) \right\}=1-TV \left( \hat{Q}_i(x),P_i(x) \right),
\end{equation}
where $\hat{Q}_i(x)$ and $\hat{q}_i(x)$ denote the distribution obtained from truncated sampling.
The change in $\hat{Q}_i(x)$ relative to $Q_i(x)$ is not deterministic and may either increase or decrease, depending on the relationship between $P_i(x)$ and $Q_i(x)$. Specifically:
\begin{itemize}[left=0pt]
    \item If the selected points of \(Q(x)\) (i.e., those with the highest \(q(x)\) values) correspond to high values of \(p(x)\), then \(\hat{\alpha}\) may increase because \(Q(x)\) concentrates probability mass on points where \(\min(p(x), \hat{q}(x))\) is large.
    \item Conversely, if the selected points of \(Q(x)\) correspond to low values of \(p(x)\), then \(\hat{\beta}_i\) may decrease because \(\min(p(x), \hat{q}(x))\) is small despite the concentration of mass.
\end{itemize} 
And
\begin{equation}
\hat{\alpha} = \mathbb{E} \left[\hat{\beta_i}\right],
\end{equation}

\subsection{Robustness of acceptance rate}
According to \eqref{eq:speedup}, the transformation from $\beta_i$ to $\hat{\beta}_i$ induced by truncated sampling on $Q(x)$ substantially affects the speedup ratio. For notational simplicity, we omit the argument $x$ and the iteration index $j$, and write $Q$ and $P$ instead of $Q_j(x)$ and $P_j(x)$ when the meaning is clear from context.

To demonstrate the robustness of the acceptance rate, we first state two lemmas that underpin the proof of the theorem presented next.

\begin{lemma}
\label{lemma:TVD}
Let $P$, $Q$, and $\hat{Q}$ be three probability distributions over the same sample space. Then the total variation distances satisfy
\begin{equation}
\label{eq:Lemma_TV}
\big|\mathrm{TV}(\hat{Q},P) - \mathrm{TV}(Q,P)\big| \le \mathrm{TV}(\hat{Q},Q),
\end{equation}
\end{lemma}

\begin{proof}
See Appendix~\ref{app:proof_TVD} for the full proof.
\end{proof}

\begin{theorem}
\label{theorem:acceptance_rate}
Let $\mathcal{V'} = \{ x^{[1]}, \dots, x^{[K]} \} \subset \mathcal{V}$ be the set of Top-$K$ tokens selected from the distribution $Q(x)$ by truncated sampling, with the total probability of the selected and remaining sets being $\rho = \sum_{v \in \mathcal{V'}} q(v)$ and $\sigma = \sum_{v \notin \mathcal{V'}} q(v)$, respectively. Let $\hat{Q}(x)$ denote the resulting distribution after truncation and renormalization over $\mathcal{V'}$. Then, the following inequality holds:
\begin{equation}
    \left| \hat{\beta} - \beta \right| \leq \sigma,
\end{equation}
where $\beta$ and $\hat{\beta}$ are the acceptance rates under the original distribution $Q(x)$ and the truncated distribution $\hat{Q}(x)$, respectively.

Besides,
\begin{equation}
\label{eq:Lemma_TVQQ}
\mathrm{TV}(\hat{Q},Q) = \sigma
\end{equation}
\end{theorem}
\begin{proof}
    See Appendix~\ref{app:proof_acceptance} for the full proof.
\end{proof}
Therefore, the expected acceptance rate is bounded as 
\begin{equation}
    |\hat{\alpha} - \alpha| \leq  \mathbb{E}\left[\sigma_i\right]
\end{equation}
where $\sigma_i = \sum_{v\notin\mathcal{V}'_i} q(v)$ denotes the truncated probability mass at the $i$-th step.

This inequality demonstrates the robustness of the acceptance rate $\beta$ under truncated sampling, 
as the deviation $|\hat{\beta}-\beta|$ is tightly controlled by the expected truncated probability mass $\sigma$. 
In particular, when the discarded probability mass is small, the acceptance rate remains close to that of the original distribution, 
indicating that truncated sampling does not substantially compromise inference quality.

\section{Multi-candidate distributed speculative decoding}
\label{sec:MU-DSD}
In SC-DSD, the draft model generates a single candidate token at each position. 
We extend this paradigm to a \textit{Multi-Candidate Distributed Speculative Decoding (MC-DSD)} setting, in which the draft model samples multiple candidate tokens per position. 
This branching process naturally forms a \textbf{token tree}, with each level corresponding to a generation step and each node representing a candidate token. 
Each path in the token tree corresponds to one candidate sequence, and the tree as a whole compactly encodes all hypotheses produced by the draft model.

\subsection{Token tree}
For simplicity of presentation, we use $x$ to denote both the node in the token tree and the token assigned to it.  And its parent node is denoted as $Par(x)$. A \textit{token tree} is defined as follows:

\begin{definition}[Token Tree]\label{def:token-tree}
A token tree $\mathcal{T}$ is a tree structure, where
each node $x \in \mathcal{T}$ is associated with a token and its parent node.  

For each node $x$, let $\bm{S}_x$ denote the token sequence corresponding to the path from the root\footnote{The root is equivalent to prefix.} to $x$. 
Formally, $\bm{S}_x$ is defined recursively as
\[
    \bm{S}_x = 
    \begin{cases}
        \{x\}, & \text{if $x$ is the root}, \\[6pt]
        \bm{S}_{Par(x)} \oplus x, & \text{otherwise},
    \end{cases}
\]
where $\oplus$ denotes sequence concatenation.
\end{definition}

\begin{definition}[Expansion-configured Token Tree]
    A token tree is \textbf{expansion-configured} with $\mathcal{K}=<k_1,\cdots,k_L>$ if every node at layer $l$ ($0\leq l< L$) has exactly $k_{l+1}$ child nodes. And $\mathcal{K}$ is called \textbf{expansion configuration}.
\end{definition}

This tree satisfies the following properties:
1). At layer 0, there is only one node, $x_{0,1}$, denoting the prefix.
2). Every node at layer $i$ ($0 \le l < L$) has exactly $k_{l+1}$ child nodes.
3). The total number of nodes at layer $l$ is 
$n_l = n_{l-1} \cdot k_l, ~ l \ge 1$.
4). The set of child nodes for the $i$-th node at layer $l$, i.e., $x_{l,i}$, is denoted by $C(x_{l,i})$. These nodes are indexed consecutively from $x_{l+1, (i-1)\cdot k_{l+1}+1}$ to $x_{l+1, i \cdot k_{l+1}}$. That is
\begin{equation}
    C(x_{l,i}) = \left\{x_{l+1,j}| j=(i-1)\cdot k_{l+1}+1,\cdots,i \cdot k_{l+1}\right\}
\end{equation}

Therefore, we have:
\begin{equation}
\begin{aligned}
    \mathcal{T} &= \bigcup_{l=0}^L \bigcup_{i=1}^{W_l} \{x_{l,i}\} = \{ x_{0,1} \} \;\cup\; \bigcup_{l=0}^{L-1} \;\bigcup_{i=1}^{W_l} \mathcal{C}(x_{l,i})\\
    &= \left \{x_{0,1}, x_{1,1},\cdots,x_{1,W_1},\cdots, x_{L,1},\cdots,x_{L,W_L} \right\}
\end{aligned}
\end{equation}
where $W_l = \prod_{j=1}^l k_j$ denotes the ``width'' of layer~$l$, i.e., the number of nodes in layer~$l$.

\begin{figure*}[htbp]
\begin{center}
\centerline{\includegraphics[width=0.9\textwidth]{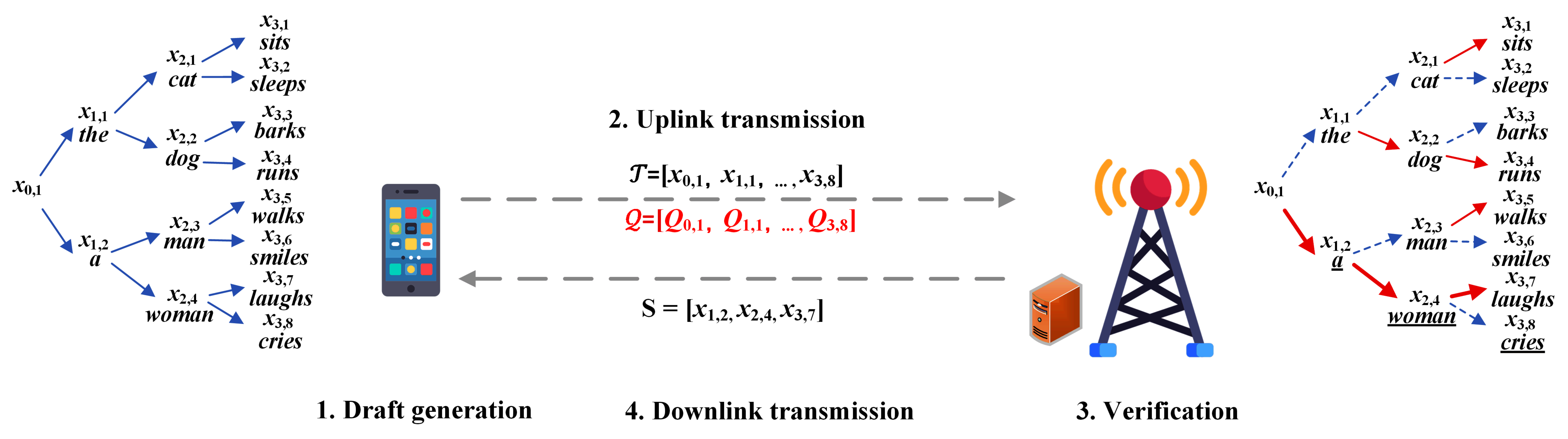}}
\caption{Multi-Candidate Distributed Speculative Decoding}
\label{fig:MU_DSD}
\end{center}
\end{figure*}
\subsection{Multi-candidate distributed speculative decoding}

Following \cite{yang2024multi}, in MC-DSD the device’s SLM generates multiple candidates 
$\{ x^{(1)}, \ldots, x^{(k)} \} \sim Q(x)$, 
which are produced according to Algorithm~\ref{alg:MCSample} and verified according to Algorithm~\ref{alg:MCVeri}. This ensures that the resulting probability distribution is identical to that obtained from standalone LLM inference, with proofs provided in \cite{yang2024multi}.

\begin{algorithm}[H]
\caption{$\textsc{MCSample}$}
\label{alg:MCSample}
\begin{algorithmic}
\STATE \text{\hspace{-1.2em} {\bfseries Input:}} $Q(x)$, $k$,
\FOR{$j=1$ \TO $k$}
    \STATE   $x^{(j)} \sim Q(x)$
\ENDFOR
\STATE \text{\hspace{-1.2em} {\bfseries Output:}} $C=\{x^{(1)},\cdots,x^{(k)}\}$
\end{algorithmic}
\end{algorithm}
\begin{algorithm}[H]
\caption{$\textsc{MCVeri}$}
\label{alg:MCVeri}
\begin{algorithmic}
\STATE \hspace{-1.2em} {\bfseries Input:} $C=\{x^{(1)},\cdots,x^{(k)}\}$, $P(x)$, $Q(x)$
\STATE $P^{(1)}(x) = P(x)$
\FOR{$i=1$ \TO $k$}
    \STATE $r \sim U[0,1]$
    \IF{
    $r < \min \left\{ 1, \frac{p^{(i)}\left( x^{(i)} \right)}{q\left( x^{(i)} \right)} \right\}$
    }
    \STATE $\tilde{x} = x^{(i)}$
    \STATE \textbf{Output}: $\tilde{x}$
    \STATE \textbf{Terminate algorithm}
    \ELSE
    \STATE 
    \refstepcounter{equation}\label{eq:resample_update}%
$p^{(i+1)}(x) := \frac{\max\{0, p^{(i)}(x)-q(x)\}}{\sum_x \max\{0, p^{(i)}(x)-q(x)\}} \; (\theequation)$
    \ENDIF
\ENDFOR
\STATE $\tilde{x} \sim P^{(k+1)}(x)$
\STATE \hspace{-1.2em} {\bfseries Output:} $\tilde{x}$
\end{algorithmic}
\end{algorithm}

The whole process is fully described in Algorithms~\ref{alg:MU-DSD}, \ref{alg:TokTreeDraft}, and \ref{alg:TokenTreeVeri}:\\ \textit{1).} The device first generates a token tree -- $\mathcal{T}$ sampled from the distribution $\mathcal{Q}$ with the expansion configuration -- $\bm{\mathcal{K}}=<k_1, \cdots, k_L>$. \textit{2).} In the uplink transmission, the device sends $\mathcal{T}$ and $\mathcal{Q}$ to the BS, which dominates most of the communication time. \textit{3).} At BS, the LLM verifies the tokens and \textit{4).} sends back the verified token sequence -- $\mathbf{S}_{veri}$ for prefix update. 
In Algorithm.~\ref{alg:TokTreeDraft}, token tree generation is autoregressive across layers, since each layer depends on the tokens from the previous one, whereas nodes within the same layer are conditionally independent and can therefore be expanded in parallel. For each parent node $x_{l, i}$, its child nodes are sampled from $Q_{m,i}(x)$ with replacement.
A toy example is illustrated in Figure~\ref{fig:MU_DSD}, with the expansion configuration given by $\bm{\mathcal{K}} = \langle 2, 2, 2 \rangle$.

The oracle time becomes
\begin{align}
    T_{oracle} &= L \cdot T_{SLM} + T_{comm} + T_{LLM} \notag\\
    &= L \cdot T_{SLM} + |\mathcal{Q}| \cdot T_{\mathcal{V}} + T_{LLM}
\end{align}
where $L$ is the number of layers for $\mathcal{T}$, and $|\mathcal{Q}|$ is the cardinality of $\mathcal{Q}$:
\begin{equation}
    |\mathcal{Q}| = 1+\sum_{l=1}^{L-1} \prod_{i=1}^{l} k_i
\end{equation}

And the speedup ratio is
\begin{equation}
    S_{inf} = \frac{1}{ \frac{L \cdot T_{SLM}+ |\mathcal{Q}| \cdot T_{\mathcal{V}}}{(L\cdot\alpha+1)\cdot T_{LLM}}  +1}
\end{equation}
While MC-DSD increases the acceptance rate by sampling additional candidates, it also incurs higher communication latency due to the need to generate and transmit a larger number of tokens and probability distributions.
Our TSLT solution can also be applied to mitigate the communication overhead in MC-DSD. The following subsection investigates its effect on the acceptance rate.
\begin{algorithm}[ht]
\caption{Multi-Candidate DSD}
\label{alg:MU-DSD}
\begin{algorithmic}
\vspace{1mm}
\STATE \hspace{-1.2em} {\bf $\triangleright$ Device (Token Tree Drafting): }
\STATE \hspace{-1.2em} {\bfseries Input:} $\bm{S}_{prefix}$, SLM -- $M_q$, $\bm{\mathcal{K}}=<k_1, \cdots, k_L>$
\STATE  $\bm{\mathcal{T}}$ and $\bm{\mathcal{Q}} \gets \textsc{\bf TokTreeDraft}(\bm{S}_{prefix}, \bm{\mathcal{K}})$  
\STATE \hspace{-1.2em} {\bfseries Output:} $\bm{\mathcal{T}}$ and $\bm{\mathcal{Q}}$
\STATE \vspace{-2mm} \hspace{-1em}\hrulefill

\STATE \hspace{-1.2em} {\bf $\triangleright$ Uplink Transmission: }
    \STATE Send $\bm{\mathcal{T}_{no-root}}$ and $\bm{\mathcal{Q}}$ to the BS via uplink transmission\footnotemark
\STATE \vspace{-2mm} \hspace{-1em}\hrulefill

\STATE \hspace{-1.2em} {\bf $\triangleright$ Edge (Token Tree Verification): }
\STATE \hspace{-1.2em} {\bfseries Input:} $\bm{S}_{prefix}$, LLM -- $M_p$, $\bm{\mathcal{K}}$, $\bm{\mathcal{T}}$, $\bm{\mathcal{Q}}$
\STATE $\bm{S}_{veri} \gets \textsc{\bf TokTreeVeri}(\bm{S}_{prefix}, \bm{\mathcal{T}}, \bm{\mathcal{Q}}, \bm{\mathcal{K}})$
\STATE \hspace{-1.2em} {\bfseries Output:} $\mathbf{S}_{veri}$
\STATE \vspace{-2mm} \hspace{-1em}\hrulefill

\STATE \hspace{-1.2em} {\bf $\triangleright$ Downlink Transmission \& Prefix Update: }
\STATE Send $\bm{S}$ back to device\\
\STATE Update $\bm{S}_{prefix} = \bm{S}_{prefix} || \bm{S}_{veri}$ at both device and edge
\end{algorithmic}
\end{algorithm}\footnotetext{$\bm{\mathcal{T}_{no-root}} = \bm{\mathcal{T}}\setminus \{x_{0,1}\}$ because prefix already exists at BS}

\begin{algorithm}[ht]
\caption{$\textsc{TokTreeDraft}$}
\label{alg:TokTreeDraft}
\begin{algorithmic}
\STATE \hspace{-1.2em} {\bfseries Input:} $\bm{S}_{prefix}$, $\bm{\mathcal{K}}$
\FOR{$l=0$ \TO $L-1$}
\FOR{$i=1$ \TO $W_l$ \textbf{(in parallel)}}
    \STATE $Q_{l,i}(x)\leftarrow M_q \left( \bm{S}_{x_{l,i}} \right)$\footnotemark
    \STATE $\mathcal{C} \left( x_{l,i} \right) \gets \textsc{\bf MCSample}(Q_{l,i}(x), k_l)$
\ENDFOR
\ENDFOR
\STATE \hspace{-1.2em} {\bfseries Output:} $\bm{\mathcal{T}} = \{x_{0,1}, x_{1,1},\cdots,x_{1,W_1},\cdots, x_{L,1},\cdots,x_{L,W_L} \}$, $\bm{\mathcal{Q}} = \{Q_{0,1}, Q_{1,1},\cdots,Q_{1,W_1},\cdots, Q_{L-1,1},\cdots,Q_{L-1,W_L-1}\}$
\end{algorithmic}
\end{algorithm}
\footnotetext{$\bm{S}_{prefix} = \bm{S}_{x_{0,1}}$}

\begin{algorithm}[ht]
\caption{$\textsc{TokenTreeVeri}$}
\label{alg:TokenTreeVeri}
\begin{algorithmic}
\STATE \hspace{-1.2em} {\bfseries Input:} $\bm{S}_{prefix}$, $\bm{\mathcal{T}}$, $\bm{\mathcal{Q}}$, $\bm{\mathcal{K}}$
\STATE $\mathcal{P} \gets M_p ( \bm{S}_{prefix} \oplus \mathcal{T})$
\STATE \hspace{-1.2em} {\bf $\triangleright$ Verify in parallel: } 
\FOR{$l=0$ \TO $L$ \textbf{and} $i=1$ \TO $W_l$ \textbf{(in parallel)}}
    \IF{$l \neq L$}
    \STATE $\tilde{x}_{l+1,i}\gets \textsc{\bf MCVeri}(C(x_{l,i}),P_{l,i}(x), Q_{l,i}(x))$
    \ELSE
    \STATE $Q_{l,i}(x)\leftarrow M_p \left( \bm{S}_{x_{l,i}} \right)$
    \STATE Sample $\tilde{x}_{l+1,i} \sim Q_{l,i}(x)$
    \ENDIF
\ENDFOR
\STATE \hspace{-1.2em} {\bf $\triangleright$ New Sequence Searching: } 
\STATE $\mathbf{S}=\mathbf{S}_{prefix}$, $l=0$, $j=1$
\FOR{$l$ \TO $L$}
    \STATE $i = j$
    \STATE $\hat{x}_{l+1} = \tilde{x}_{l+1,i}$
    \STATE $\mathbf{S} = \mathbf{S} \oplus \hat{x}_{l+1}$
    \IF{ $\tilde{x}_{l+1,i} \in C(x_{l,i})$ }
        \STATE $j = \text{pos}\left(\tilde{x}_{l+1,i}|C(x_{l,i})\right)$\footnotemark
        \STATE $l=l+1$
    \ELSE
        \STATE {\bfseries Output:} $\mathbf{S}_{veri} = \hat{x}_1 \oplus...\oplus \hat{x}_{l+1}$
        \STATE \textbf{Terminate algorithm}
    \ENDIF
\ENDFOR
    \STATE $\hat{x}_{l+1} = \tilde{x}_{l+1,j}$
\STATE \hspace{-1.2em} {\bfseries Output:} $\mathbf{S}_{veri} = \hat{x}_1 \oplus...\oplus \hat{x}_{l+1}$
\end{algorithmic}
\end{algorithm}

\subsection{Acceptance rate in MC-DSD}
\begin{definition}
    For multi-candidate speculative decoding with $k$ candidates, we define $\beta^{(j)}$ as the acceptance rate of the $j$-th candidate conditioned on that candidates $x^{(1)},\cdots,x^{(j-1)}$ are rejected:
    \begin{equation}
        \beta^{(j)} = Pr \left[ x^{(j)} \text{accepted} \mid x^{(1)},\cdots,x^{(j-1)} \text{rejected} \right].
    \end{equation}
\end{definition}
Therefore, we have
\begin{lemma}
    The acceptance rate of the $j$-th candidate token is 
    \begin{equation}
        \beta^{(j)} =  \sum_{x \in \mathcal{V} } \min \left\{ q(x),p^{(j)}(x) \right\} = 1-TV(Q,P^{(j)})
    \end{equation}
\end{lemma}
\begin{proof}
    According to Algorithm~\ref{alg:MCVeri}, we have
    \begin{equation}
        \beta^{(j)} = \sum_{x \in \mathcal{V} } q(x) \min \left\{ 1,\frac{p^{(j)}(x)}{q(x)} \right\} = \sum_{x \in \mathcal{V} } \min \left\{ q(x),p^{(j)}(x) \right\}
    \end{equation}
\end{proof}

We first establish several useful lemmas and corollary that will be used to prove Theorem~\ref{thm:acceptance_MC-DSD}.
\begin{lemma}
\label{lemma:TV_MC-DSD}
Let  $P$, $Q$, $\hat{P}$ and $\hat{Q}$ be four probability distributions over the same sample space. Then the total variation distances satisfy
\begin{equation}
\label{eq:Lemma_TV_MC-DSD}
\big|\mathrm{TV}(\hat{Q},\hat{P}) - \mathrm{TV}(Q,P)\big| \le \mathrm{TV}(\hat{Q},Q) + \mathrm{TV}(\hat{P},P),
\end{equation}
\end{lemma}
\begin{proof}
    See Appendix~\ref{app:proof_TV_MC-DSD} for the full proof.
\end{proof}

Let $P^{(i)}(x)$ and $Q(x)$ denote the probability distributions in \eqref{eq:resample_update} before applying truncated sampling, and let $\hat{P}^{(i)}(x)$ and $\hat{Q}(x)$ denote the corresponding distributions after the TSLT solution Then we have
\begin{lemma}\label{lemma_TV_recurse}
The total variance between $P^{i}(x)$ and $\hat{P}^{i}(x)$ is bounded by
\begin{equation}\label{eq:TV_PP_Bound}
    TV(P^{(i)},\hat{P}^{(i)}) \leq \frac{2}{Z^{(i)}}\left[ TV(P^{(i-1)}, \hat{P}^{(i-1)}) + TV(Q,\hat{Q})\right],
\end{equation}
where $\hat{Z}^{(i)} =TV(P^{(i-1)},\hat{Q})$.
\end{lemma}
\begin{proof}
     See Appendix~\ref{app:proof_TV_recurse} for the full proof.
\end{proof}
\begin{corollary}
\label{cor:1}
For $i=2$
\begin{equation}\label{eq:TV_PP2_Bound}
    TV(P^{(2)},\hat{P}^{(2)}) \leq \frac{2}{Z^{(1)}} TV(Q,\hat{Q}).
\end{equation}
\end{corollary}

With Lemma~\ref{lemma:TV_MC-DSD} and Lemma~\ref{lemma_TV_recurse} established, we obtain the following theorem.
\begin{theorem}\label{thm:acceptance_MC-DSD}
    Let $\mathcal{V'} = \{ x^{[1]}, \dots, x^{[K]} \} \subset \mathcal{V}$ be the set of Top-$K$ tokens selected from the distribution $Q(x)$ by truncated sampling, with the total probability of the selected and remaining sets being $\rho = \sum_{v \in \mathcal{V'}} q(v)$ and $\sigma = \sum_{v \notin \mathcal{V'}} q(v)$, respectively. Let $\hat{Q}(x)$ denote the resulting distribution after truncation and renormalization over $\mathcal{V'}$. Then, the following inequality holds:
\begin{equation}
    \left| \hat{\beta}^{(j)} - \beta^{(j)} \right| \leq \left[\sum_{k=1}^{j-1} \left( \prod_{s=k+1}^{j} \frac{2}{Z^{(s)}} \right)+1 \right]\sigma,
\end{equation}
where $\beta^{(j)}$ and $\hat{\beta}^{(j)}$ are the acceptance rates under the original distribution $Q(x)$ and the truncated distribution $\hat{Q}(x)$ for the $j$-th candidate, respectively.
\end{theorem}
\begin{proof}
Combing Lemma~\ref{lemma_TV_recurse} and \ref{cor:1}, we have
\begin{equation}
TV(P^{(j)},\hat{P}^{(j)}) \le 
\sum_{k=1}^{j-1} \left( \prod_{s=k+1}^{j} \frac{2}{Z^{(s)}} \right) \sigma.
\end{equation}
Then, 
\begin{align}\label{eq:bound_TV_MC}
     \left| \hat{\beta}^{(j)} - \beta^{(j)} \right|
    =& | TV(\hat{Q},\hat{P}^{(j)}) - TV(Q,P^{(j)}) |\notag\\
    \leq& \mathrm{TV}(\hat{Q},Q) + \mathrm{TV}(\hat{P}^{(j)},P^{(j)}) \notag\\
    \leq& \left[\sum_{k=1}^{j-1} \left( \prod_{s=k+1}^{j} \frac{2}{Z^{(s)}} \right)+1 \right]\sigma
\end{align}
\end{proof}

This inequality demonstrates the robustness of the acceptance rate $\hat{\beta}^{(j)}$ under truncated sampling, as the deviation $|\hat{\beta}^{(j)}-\beta^{(j)}|$ is tightly controlled by the expected truncated probability mass in \eqref{eq:bound_TV_MC}. 

Then the expected acceptance rate for the $j$-th candidate is
\begin{equation}
    \alpha^{(j)} = \mathbb{E}[\beta^{(j)}]
\end{equation}

Total acceptance rate is 
\begin{align}
\theta = \beta^{(1)}+\sum_{i=2}^{k} \beta^{(i)} \prod_{j=1}^{i-1} (1-\beta^{(j)}),
\end{align}

And the final acceptance rate is
\begin{align}
    \alpha = \mathbb{E}[\theta] \approx \alpha^{(1)}+\sum_{i=2}^{k} \alpha^{(i)} \prod_{j=1}^{i-1} (1-\alpha^{(j)}),
\end{align}

\section{Experiment}
\label{sec:experiment}
We deploy two models on separate NVIDIA A800 80GB GPUs: a \textbf{68M-Llama} model for draft token generation and a \textbf{7B-Llama} model for verification. The vocabulary size is set to $32$K~\cite{touvron2023llama} with FP16, i.e., $D_{\mathcal{V}} = 0.5$~Mbit. For truncated sampling, we employ a Top-$K$ strategy. The expansion configuration is specified as $\bm{\mathcal{K}} = \langle 2, 2, 2 \rangle$, indicating that up to two candidate tokens are considered at each speculative decoding step. 
Table~\ref{tab:exp_settings} summarizes the experimental setup.
\begin{table}[htbp]
\centering
\caption{Summary of experimental settings.}
\label{tab:exp_settings}
\begin{tabular}{l l}
\toprule
\textbf{Parameter} & \textbf{Value} \\
\midrule
SLM & 68M-Llama \\
LLM & 7B-Llama \\
GPU & 2~NVIDIA A800 80GB \\
Vocabulary size & 32K~\cite{touvron2023llama} \\
Precision & FP16 ($D_{\mathcal{V}} = 0.5$~Mbit) \\
Truncated sampling & Top-$K$ \\
Speculative decoding expansion & $\bm{\mathcal{K}} = \langle 2, 2, 2 \rangle$ \\
\bottomrule
\end{tabular}
\end{table}

We first compare the acceptance rate and the retained probability mass under different truncated set sizes ($K$). Next, we evaluate the inference latency to quantify the efficiency gains of the Truncated Sparse Logits Transmission (TSLT) scheme. Specifically, we consider the inference latency under varying $K$ and uplink transmission rates.

The total inference time for distributed speculative decoding (DSD) can be expressed as:
\begin{align}
    T_{DSD} = T_{comp} + N_{\text{oracle}} \cdot T_{comm},
\end{align}
where $T_{comp}$ denotes the computation time measured directly on the target hardware, and $N_{\text{oracle}}$ is the number of oracle calls for a given task. The communication time per oracle, $T_{comm}$, is estimated via simulation. For Vanilla DSD, $T_{comm} = L \cdot D_{\mathcal{V}} / R_{up}$, while for MC-DSD, $T_{comm} = |\mathcal{Q}| \cdot D_{\mathcal{V}} / R_{up}$, where $D_{\mathcal{V}}$ is the size of the transmitted logits and $R_{up}$ is the uplink transmission rate.




\subsection{Top-$K$ probability mass}
To evaluate how much of the original distribution is preserved under truncated sampling, we first examine the retained probability mass across different truncated set sizes ($K$).
\begin{figure}
    \centering
    \includegraphics[width=\linewidth]{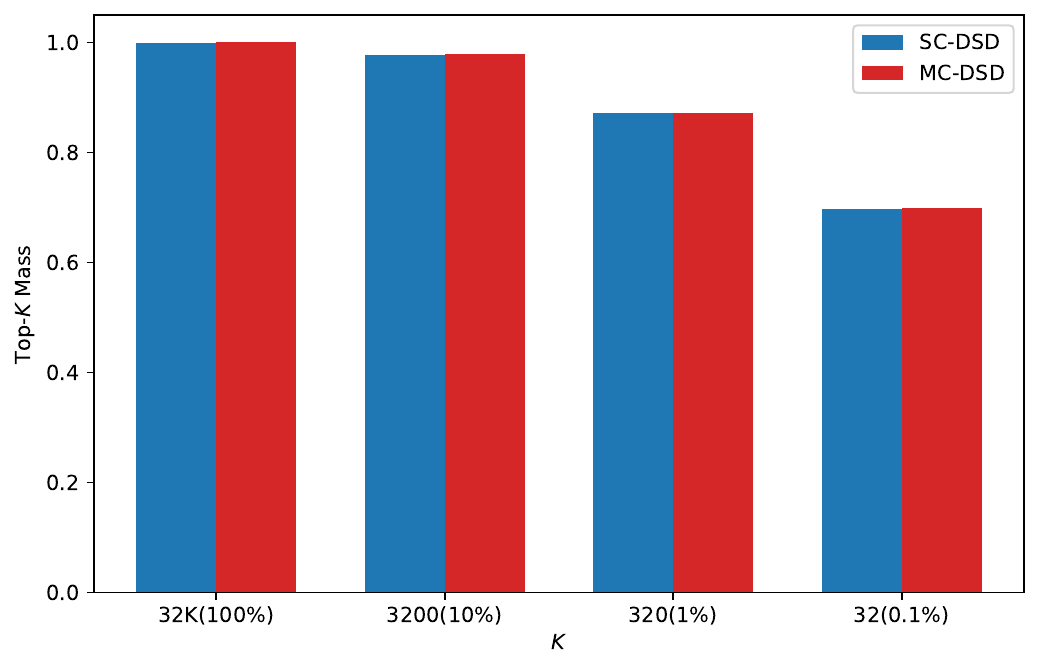}
    \caption{Top-$K$ probability mass different $K$}
    \label{fig:probability_mass}
\end{figure}

Figure~\ref{fig:probability_mass} illustrates the retained probability mass under different truncated set sizes ($K$). As observed, when $K=3{,}200$ (10\% of the vocabulary), the probability mass remains almost unchanged, indicating that the majority of the distribution is captured. Even when $K$ is reduced to $320$ (about 1\%), the retained probability mass stays relatively high, around 0.85, demonstrating that a small truncated set is sufficient to preserve most of the model's predictive distribution. This highlights the effectiveness of the Truncated Sparse Logits Transmission (TSLT) scheme in significantly reducing communication overhead while maintaining inference fidelity.

\subsection{Acceptance rate}

In this subsection, we analyze the acceptance rate of different decoding strategies under varying Top-$K$ settings, which serves as a key indicator of the effectiveness and robustness of our proposed TSLT solution.
\begin{figure}
    \centering
    \includegraphics[width=\linewidth]{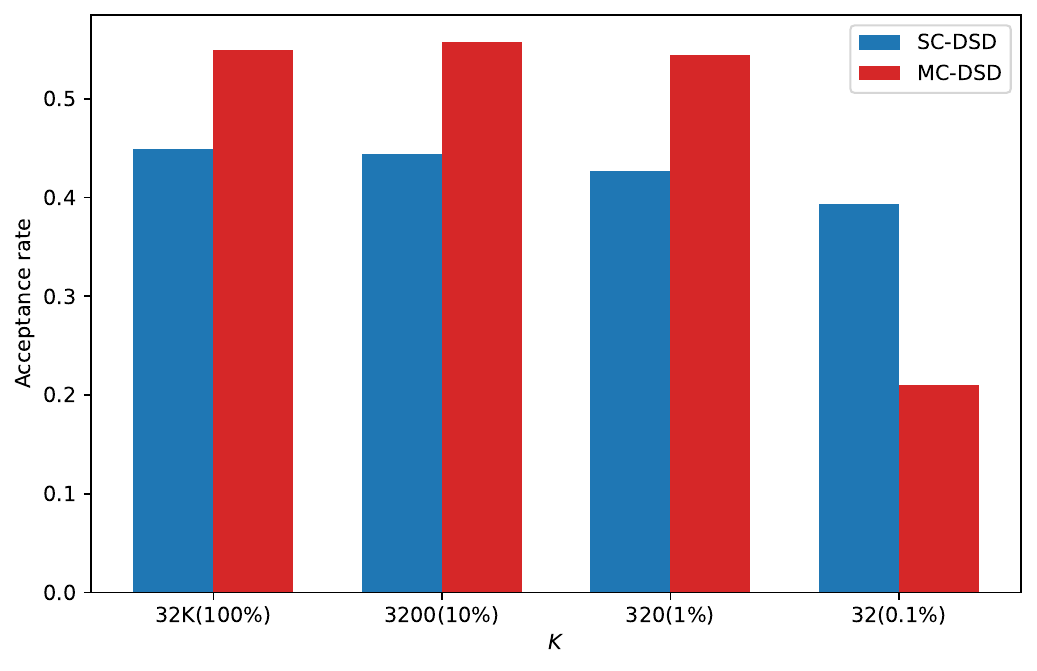}
    \caption{Acceptance rate under different $K$}
    \label{fig:acceptance_rate}
\end{figure}

As shown in Figure~\ref{fig:acceptance_rate}, we compare the acceptance rates under both SC-DSD and MC-DSD scenarios. As $K$ decreases from $32,000$ to $3,200$, the acceptance rate exhibits only a minor decline, indicating that our proposed TSLT solution is robust: transmitting fewer logits introduces only a small approximation error and does not significantly degrade acceptance performance. Even when only $1\%$ of the logits are transmitted, the performance remains largely unaffected. However, for $K=32$, the acceptance rate drops noticeably. This behavior corresponds to the top-K mass shown in Figure~\ref{fig:probability_mass}; when $K=32$, the residual probability (i.e., $1-$top-K mass) becomes significant, leading to a more pronounced decrease in acceptance performance.

Furthermore, when comparing the two approaches, MC-DSD consistently achieves a higher acceptance rate than SC-DSD. This improvement can be attributed to the multi-candidate sampling strategy employed in MC-DSD, where an additional candidate is sampled in each iteration, thereby increasing the likelihood of acceptance. Overall, these results highlight that our TSLT-enhanced MC-DSD framework not only maintains stability with smaller $K$ values but also benefits from enhanced acceptance due to more extensive candidate exploration.

\subsection{Speedup ratio}
To evaluate the efficiency gains of the Truncated Sparse Logits Transmission (TSLT) scheme, we examine the inference speedup achieved under varying uplink transmission rates. This analysis helps quantify how TSLT mitigates communication overhead and accelerates DSD inference.

\begin{figure}
    \centering
    \includegraphics[width=\linewidth]{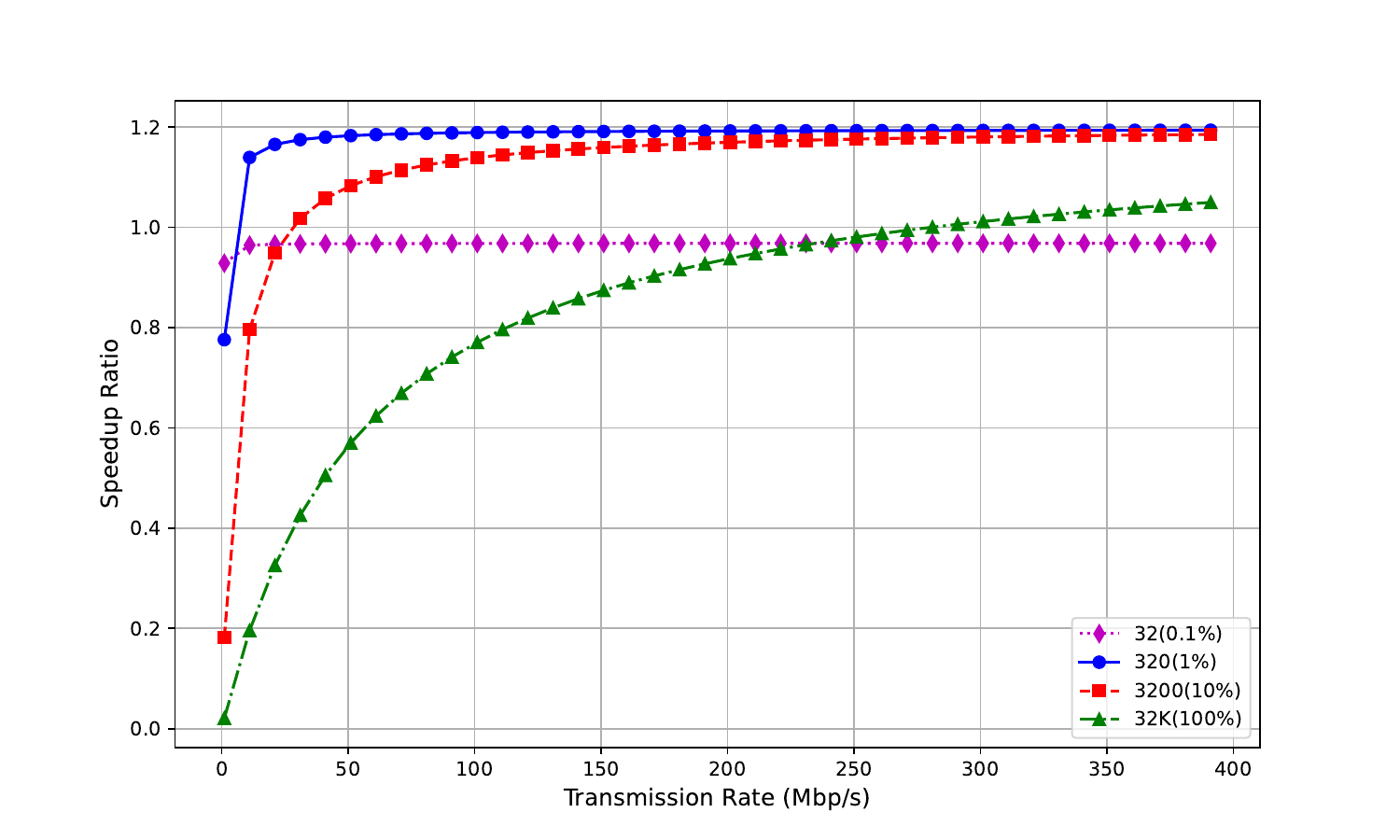}
    \caption{Speedup Ratio under different uplink transmission rates (SC-DSD).}
    \label{fig:Speedup Ratio}
\end{figure}

\begin{figure}
    \centering
    \includegraphics[width=\linewidth]{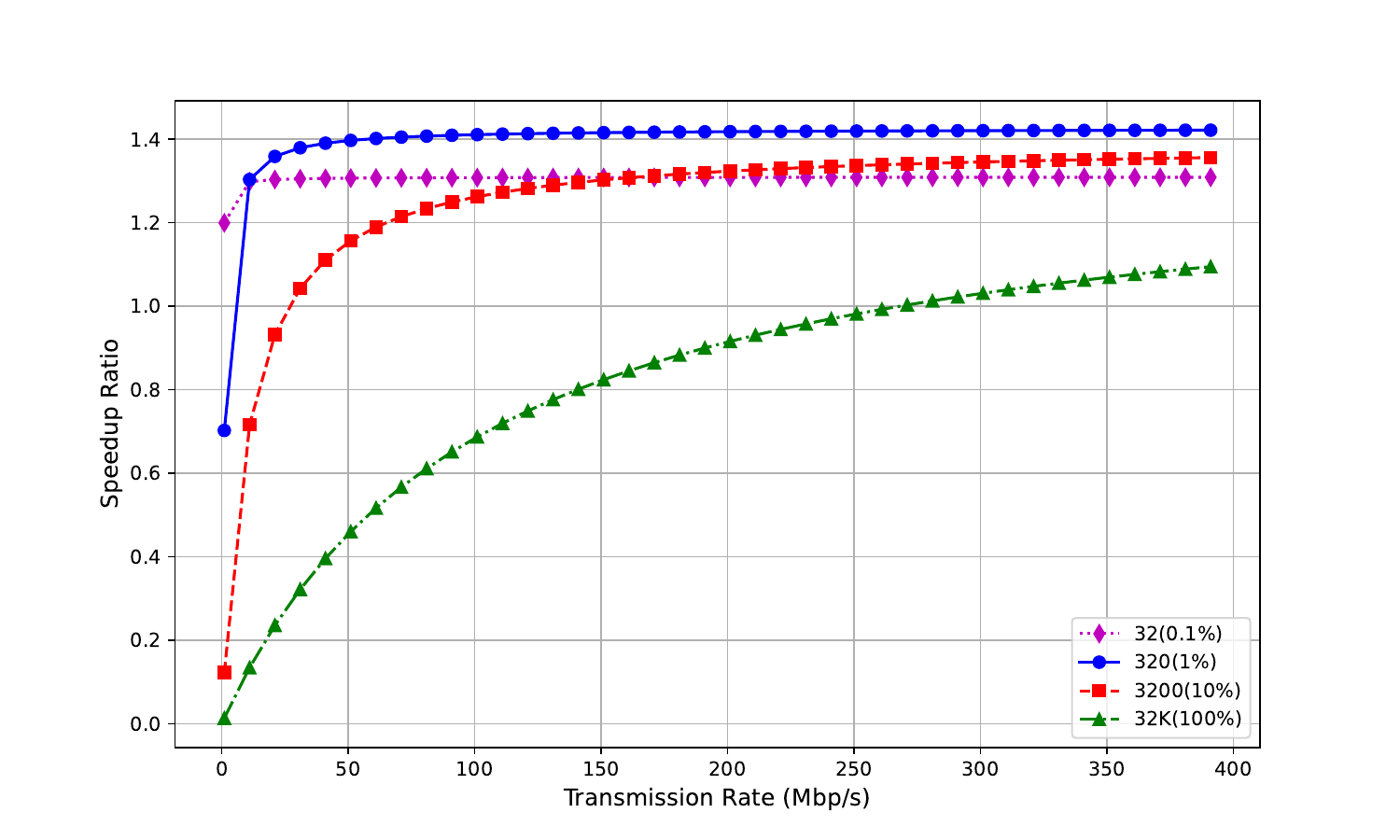}
    \caption{Speedup Ratio under different uplink transmission rates  (MC-DSD).}
    \label{fig:Speedup Ratio2}
\end{figure}

Figures~\ref{fig:Speedup Ratio} and~\ref{fig:Speedup Ratio2} depict the speedup ratio under different uplink rates for SC-DSD and MC-DSD, respectively. As observed, the use of TSLT significantly improves the speedup, demonstrating its effectiveness in mitigating communication bottlenecks and enhancing inference efficiency under constrained network conditions. Furthermore, MC-DSD consistently outperforms SC-DSD, as generating multiple draft candidates per token position allows for greater parallelism and higher overall acceleration.

However, in the extreme compression scenario where $K=32$ (i.e., 0.1\% compression), the speedup ratio decreases. This is due to a lower acceptance rate, which increases the number of calls to the SLM and LLM—leading to longer computation time. This effect is particularly pronounced at higher uplink rates, where the communication bottleneck is less significant and computation dominates the overall latency.

\section{Conclusion and Future Work}
\label{sec:conclusion}
In this work, we investigated speculative decoding for distributed large language model (LLM) inference in AI-native radio access networks (AI-RAN) and proposed the Truncated Sparse Logits Transmission (TSLT) solution to mitigate the high uplink communication cost. We rigorously established the robustness of the acceptance rate under TSLT through a series of theorems and proofs, and further extended it to multi-candidate speculative decoding to verify multiple draft tokens simultaneously. Experimental results validate our theoretical findings, demonstrating that TSLT can significantly reduce communication overhead while preserving inference performance, highlighting its potential for efficient and scalable distributed LLM inference in resource-constrained AI-RAN scenarios.

Our work assumes that the device-based small language model (SLM) and server-based LLM share the same vocabulary, and we focus on communication efficiency without considering computational overhead on either side. Addressing heterogeneous vocabularies—e.g., by integrating strategies from Lossless Speculative Decoding for Heterogeneous Vocabularies or OmniDraft—and optimizing the trade-off between communication and computation are important directions for future work, which can further enhance the practicality and efficiency of distributed LLM inference in AI-RAN.

\section*{Appendix}
\appendix

\section{Proof of Lemma~\ref{lemma:TVD}}
\label{app:proof_TVD}
\begin{proof} 
Let $a(v)=\hat{q}(v)-p(v)$ and $b(v)=q(v) - p(v)$, we have
\begin{align}
    \label{eq:TV_QP}
    \mathrm{TV}(\hat{Q},P) = \frac{1}{2} \sum_{v\in\mathcal{V}} |a(v)|, \\
    \label{eq:TV_QQP}
    \mathrm{TV}(Q,P) = \frac{1}{2} \sum_{v\in\mathcal{V}} |b(v)|
\end{align}
\begin{equation}
\label{eq:TV_QQQ}
    \mathrm{TV}(\hat{Q},Q)=\frac{1}{2}\sum_{v\in\mathcal{V}} |\hat{q}(v)-q(v)|
    = \frac{1}{2}\sum_{v\in\mathcal{V}}|a(v) - b(v)|, 
\end{equation}
By the reverse triangle inequality,
\begin{equation}
    \left| \frac{1}{2}\sum_{v\in\mathcal{V}}|a(v)|-\frac{1}{2}\sum_{v\in\mathcal{V}}|b(v)| \right| \leq \frac{1}{2}\sum_{v\in\mathcal{V}}|a(v) - b(v)|
\end{equation}
Combing this with \eqref{eq:TV_QP}, \eqref{eq:TV_QQP} and \eqref{eq:TV_QQQ}, we obtain
\begin{equation}
\big|\mathrm{TV}(\hat{Q},P) - \mathrm{TV}(Q,P)\big| \le \mathrm{TV}(\hat{Q},Q).
\end{equation}
\end{proof}

\section{Proof of Theorem~\ref{theorem:acceptance_rate}}
\label{app:proof_acceptance}
\begin{proof}
Let $\mathcal{V'} =  \left\{x^{[1]},\cdots,x^{[K]} \right\}$, and $\sum_{v\in\mathcal{V'}}q(v)=\rho$, the sparsified distribution after truncated sampling is
\begin{equation}
    \hat{q}(x) = \begin{cases}
    0,  & x < 0, \\
    \frac{q(x)}{\rho},  & x \geq 0.
    \end{cases}
\end{equation}
Therefore, we have
\begin{align}
&\sum_{v\in \mathcal{V}}\left| \hat{q}(v)-q(v)   \right|_1 \notag\\
=& \sum_{v \in \mathcal{V'}} \left| \hat{q}(v)-q(v) \right|_1 +  \sum_{v \notin \mathcal{V'}} \left| \hat{q}(v)-q(v) \right|_1 \notag\\
\overset{(a)}{=}&  \sum_{v \in \mathcal{V'}} \left| \frac{q(v)}{\rho}-q(v) \right|_1 +  \sum_{v \notin \mathcal{V'}} q(v) \notag \\
=& \frac{1-\rho}{\rho} \sum_{v \in \mathcal{V'}} q(v)   +  \sigma \notag \\
\overset{(b)}{=}& 2\sigma
\end{align}
where $(a)$ is from \eqref{eq:truncation}, and $(b)$ is from $\rho = \sum_{v \in \mathcal{V'}} q(v)$ and $\sigma = \sum_{v \notin \mathcal{V'}} q(v)$.

The total variance distance (TVD) between $Q(x)$ and $\hat{Q}(x)$ is 
\begin{align}
    TV \left( \hat{Q}(x),Q(x) \right)
    = \frac{1}{2} \sum_{v\in \mathcal{V}}\left| \hat{q}(v)-q(v) \right|_1 = \sigma
\end{align}
And
\begin{align}
    \left| \hat{\beta}- \beta \right| 
    \overset{(c)}{=}& \left| TV\left( \hat{Q}(x),P(x) \right)-TV\left( Q(x),P(x) \right) \right| \notag\\
    \leq& ~TV\left( \hat{Q}(x),Q(x) \right) = \sigma.
\end{align}
where $(c)$ follows from the following lemma.
\end{proof}

\section{Proof of Lemma~\ref{lemma:TV_MC-DSD}}
\label{app:proof_TV_MC-DSD}
\begin{proof}
\begin{align}
    & |q(x)-p(x)| \notag\\
    =& \big| [q(x)-\hat{q}(x)]+[\hat{q}(x)-\hat{p}(x)]+[\hat{p}(x)-p(x)] \big| \notag\\
    \leq& |q(x)-\hat{q}(x)|+|\hat{q}(x)-\hat{p}(x)|+|\hat{p}(x)-p(x)|
\end{align}
Therefore, we have
\begin{equation}
    TV(Q,P) \leq TV(Q,\hat{Q}) + TV(\hat{Q},\hat{P}) + TV(\hat{P},P)
\end{equation}
And 
\begin{equation}
    TV(Q,P)-TV(\hat{Q},\hat{P}) \leq TV(Q,\hat{Q}) +  TV(\hat{P},P)
\end{equation}
By symmetry, interchanging $(Q,P)$ and $(\hat{Q},\hat{P})$ leads to the reverse inequality; 
combining the two inequalities establishes the desired absolute bound.
\end{proof}
\section{Proof of Lemma~\ref{lemma_TV_recurse}}
\label{app:proof_TV_recurse}
\begin{proof}
Let 
\begin{align}
    r^{(i)}(x) = \max \{0, p^{(i-1)}(x)-q(x)\} \\
    \hat{r}^{(i)}(x) = \max \{0, \hat{p}^{(i-1)}(x)-\hat{q}(x)\}
\end{align}
And
\begin{align}
    p^{(i)}(x) = \frac{r^{(i)}(x)}{Z^{(i)}}, ~~~Z^{(i)} = \sum_{x \in \mathcal{V}} r^{(i)}(x) = TV(P^{(i)},Q)\\
    \hat{p}^{(i)}(x) = \frac{\hat{r}^{(i)}(x)}{\hat{Z}^{(i)}}, ~~~\hat{Z}^{(i)} = \sum_{x \in \mathcal{V}} \hat{r}^{(i)}(x) = TV(\hat{P}^{(i)},\hat{Q}) \label{eq:Z_hat}
\end{align}
Then
\begin{align}
\label{eq:ppp}
    & \sum_{x\in \mathcal{V}} | p^{(i)}(x) -\hat{p}^{(i)}(x) | \notag\\
    =& \sum_{x\in \mathcal{V}} \left|  \frac{r^{(i)}(x)}{Z^{(i)}} - \frac{\hat{r}^{(i)}(x)}{\hat{Z}^{(i)}}   \right| \notag\\
    =& \sum_{x\in \mathcal{V}} \left| \frac{r^{(i)}(x) - \hat{r}^{(i)}(x)}{Z^{(i)}} + \hat{r}^{(i)}(x)\left(\frac{1}{\hat{Z}^{(i)}} - \frac{1}{Z^{(i)}}\right) \right| \notag\\
    \overset{(a)}{\leq}& \sum_{x\in \mathcal{V}} \left| \frac{r^{(i)}(x) - \hat{r}^{(i)}(x)}{Z^{(i)}} \right| + \sum_{x\in \mathcal{V}} \left| \hat{r}^{(i)}(x)\left(\frac{1}{\hat{Z}^{(i)}} - \frac{1}{Z^{(i)}}\right) \right| \notag\\
    =& \sum_{x\in \mathcal{V}} \frac{ \left| r^{(i)}(x) - \hat{r}^{(i)}(x) \right|}{Z^{(i)}}  +  \frac{\left| Z^{(i)}-\hat{Z}^{(i)} \right|}{\hat{Z}^{(i)}Z^{(i)}} \sum_{x\in \mathcal{V}} \hat{r}^{(i)}(x) \notag\\
    \overset{(b)}{=}& \frac{ \sum_{x\in \mathcal{V}} \left| r^{(i)}(x) - \hat{r}^{(i)}(x) \right|}{Z^{(i)}}  +  \frac{\left| Z^{(i)}-\hat{Z}^{(i)} \right|}{Z^{(i)}} \notag\\
    \overset{(c)}{\leq}& \frac{2}{Z^{(i)}} \sum_{x\in \mathcal{V}} \left| r^{(i)}(x) - \hat{r}^{(i)}(x) \right| \notag\\
    \overset{(d)}{\leq}& \frac{2}{Z^{(i)}} \left[ \sum_{x\in \mathcal{V}}|p^{(i-1)}(x)-\hat{p}^{(i-1)}(x)| + \sum_{x\in \mathcal{V}}|\hat{q}(x)-q(x)|\right]
\end{align}
where (a) is from triangular inequality, (b) is from \eqref{eq:Z_hat}, and (c) is from the followings:
\begin{align}
\left| \sum_{x\in \mathcal{V}}r^{(i)}(x) -  \sum_{x\in \mathcal{V}}\hat{r}^{(i)}(x) \right| 
=& \left| \sum_{x\in \mathcal{V}} \left[r^{(i)}(x) - \hat{r}^{(i)}(x) \right]\right| \notag\\
\overset{(e)}{\leq}& \sum_{x\in \mathcal{V}} \left| r^{(i)}(x) - \hat{r}^{(i)}(x) \right| \notag\\
\end{align} 
and (e) is according to the inverse triangular inequality.
For (d), As ReLU is 1-Lipschitz, the following holds: 
\begin{align}
    &|r^{(i)}(x)-\hat{r}^{(i)}(x)| \notag\\
    \leq& |r^{(i)}(x)-\hat{r}^{(i)}(x)| \notag\\
    =& |[p^{(i-1)}(x)-q(x)] - [\hat{p}^{(i-1)}(x)-\hat{q}(x)]| \notag\\
    =& |[p^{(i-1)}(x)-\hat{p}^{(i-1)}(x)] + [\hat{q}(x)-q(x)]| \notag\\
    \leq& |p^{(i-1)}(x)-\hat{p}^{(i-1)}(x)| + |\hat{q}(x)-q(x)|
\end{align}

Dividing both sides of \eqref{eq:ppp} by two yields \eqref{eq:TV_PP_Bound}. 
Specifically, $TV \big( P^{(1)}, \hat{P}^{(1)} \big) = 0$, since $P^{(1)} = \hat{P}^{(1)}$. Hence, \eqref{eq:TV_PP_Bound} holds.

\end{proof}

\section*{Declaration of generative AI and AI-assisted technologies in the manuscript preparation process}

During the preparation of this work the author(s) used ChatGPT (OpenAI) in order to improve the clarity and readability of the manuscript. After using this tool, the author(s) reviewed and edited the content as needed and take full responsibility for the content of the published article.

\printcredits

\bibliographystyle{cas-model2-names}

\bibliography{cas-refs}

@article{ding2024hybrid,
  title={Hybrid {LLM}: Cost-efficient and quality-aware query routing},
  author={Ding, Dujian and Mallick, Ankur and Wang, Chi and Sim, Robert and Mukherjee, Subhabrata and Ruhle, Victor and Lakshmanan, Laks VS and Awadallah, Ahmed Hassan},
  journal={arXiv preprint arXiv:2404.14618},
  year={2024}
}

@inproceedings{hao2024hybrid,
  title={Hybrid {SLM} and {LLM} for edge-cloud collaborative inference},
  author={Hao, Zixu and Jiang, Huiqiang and Jiang, Shiqi and Ren, Ju and Cao, Ting},
  booktitle={Proceedings of the Workshop on Edge and Mobile Foundation Models},
  pages={36--41},
  year={2024}
}

@article{chen2023accelerating,
  title={Accelerating large language model decoding with speculative sampling},
  author={Chen, Charlie and Borgeaud, Sebastian and Irving, Geoffrey and Lespiau, Jean-Baptiste and Sifre, Laurent and Jumper, John},
  journal={arXiv preprint arXiv:2302.01318},
  year={2023}
}

@inproceedings{leviathan2023fast,
  title={Fast inference from transformers via speculative decoding},
  author={Leviathan, Yaniv and Kalman, Matan and Matias, Yossi},
  booktitle={International Conference on Machine Learning},
  pages={19274--19286},
  year={2023},
  organization={PMLR}
}

@inproceedings{zhao2024edge,
  title={Edge and terminal cooperation enabled {LLM} deployment optimization in wireless network},
  author={Zhao, Wentao and Jing, Wenpeng and Lu, Zhaoming and Wen, Xiangming},
  booktitle={2024 IEEE/CIC International Conference on Communications in China (ICCC Workshops)},
  pages={220--225},
  year={2024},
  organization={IEEE}
}

@article{oh2025communication,
  title={Communication-Efficient Hybrid Language Model via Uncertainty-Aware Opportunistic and Compressed Transmission},
  author={Oh, Seungeun and Kim, Jinhyuk and Park, Jihong and Ko, Seung-Woo and Choi, Jinho and Quek, Tony QS and Kim, Seong-Lyun},
  journal={arXiv preprint arXiv:2505.11788},
  year={2025}
}

@article{xie2025novel,
  title={A Novel Hat-Shaped Device-Cloud Collaborative Inference Framework for Large Language Models},
  author={Xie, Zuan and Xu, Yang and Xu, Hongli and Liao, Yunming and Yao, Zhiwei},
  journal={arXiv preprint arXiv:2503.18989},
  year={2025}
}

@article{shao2025ai,
  title={{AI} Flow at the Network Edge},
  author={Shao, Jiawei and Li, Xuelong},
  journal={IEEE Network},
  year={2025},
  publisher={IEEE}
}

@article{oh2024uncertainty,
  title={Uncertainty-Aware Hybrid Inference with On-Device Small and Remote Large Language Models},
  author={Oh, Seungeun and Kim, Jinhyuk and Park, Jihong and Ko, Seung-Woo and Quek, Tony QS and Kim, Seong-Lyun},
  journal={arXiv preprint arXiv:2412.12687},
  year={2024}
}

@inproceedings{park2025uncertainty,
  title={Uncertainty-Aware Opportunistic Hybrid Language Model in Wireless Robotic Systems},
  author={Park, Jeyoung and Lim, Yeonsub and Oh, Seungeun and Park, Jihong and Kim, Seong-Lyun},
  booktitle={ICML 2025 Workshop on Machine Learning for Wireless Communication and Networks (ML4Wireless)},
  pages = {--}
}

@article{park2025action,
  title={Action Deviation-Aware Inference for Low-Latency Wireless Robots},
  author={Park, Jeyoung and Lim, Yeonsub and Oh, Seungeun and Park, Jihong and Choi, Jinho and Kim, Seong-Lyun},
  journal={arXiv preprint arXiv:2510.02851},
  year={2025}
}

@inproceedings{ning2025dssd,
  title={{DSSD}: Efficient Edge-Device Deployment and Collaborative Inference via Distributed Split Speculative Decoding},
  author={NING, JIAHONG and ZHENG, Ce and Yang, Tingting},
  booktitle={ICML 2025 Workshop on Machine Learning for Wireless Communication and Networks (ML4Wireless)},
  pages = {--}
}

@article{zheng2025communication,
  title={Communication-Efficient Collaborative LLM Inference via Distributed Speculative Decoding},
  author={Zheng, Ce and Yang, Tingting},
  journal={arXiv preprint arXiv:2509.04576},
  year={2025}
}

@article{khisti2024multi,
  title={Multi-Draft Speculative Sampling: Canonical Decomposition and Theoretical Limits},
  author={Khisti, Ashish and Ebrahimi, M Reza and Dbouk, Hassan and Behboodi, Arash and Memisevic, Roland and Louizos, Christos},
  journal={arXiv preprint arXiv:2410.18234},
  year={2024}
}

@article{yang2024multi,
  title={Multi-candidate speculative decoding},
  author={Yang, Sen and Huang, Shujian and Dai, Xinyu and Chen, Jiajun},
  journal={arXiv preprint arXiv:2401.06706},
  year={2024}
}

@inproceedings{miao2024specinfer,
  title={Specinfer: Accelerating large language model serving with tree-based speculative inference and verification},
  author={Miao, Xupeng and Oliaro, Gabriele and Zhang, Zhihao and Cheng, Xinhao and Wang, Zeyu and Zhang, Zhengxin and Wong, Rae Ying Yee and Zhu, Alan and Yang, Lijie and Shi, Xiaoxiang and others},
  booktitle={Proceedings of the 29th ACM International Conference on Architectural Support for Programming Languages and Operating Systems, Volume 3},
  pages={932--949},
  year={2024}
}

@article{lu2024improving,
  title={Improving multi-candidate speculative decoding},
  author={Lu, Xiaofan and Zeng, Yixiao and Ma, Feiyang and Yu, Zixu and Levorato, Marco},
  journal={arXiv preprint arXiv:2409.10644},
  year={2024}
}

@article{yin2024theoretical,
  title={A theoretical perspective for speculative decoding algorithm},
  author={Yin, Ming and Chen, Minshuo and Huang, Kaixuan and Wang, Mengdi},
  journal={Advances in Neural Information Processing Systems},
  volume={37},
  pages={128082--128117},
  year={2024}
}

@article{touvron2023llama,
  title={Llama 2: Open foundation and fine-tuned chat models},
  author={Touvron, Hugo and Martin, Louis and Stone, Kevin and Albert, Peter and Almahairi, Amjad and Babaei, Yasmine and Bashlykov, Nikolay and Batra, Soumya and Bhargava, Prajjwal and Bhosale, Shruti and others},
  journal={arXiv preprint arXiv:2307.09288},
  year={2023}
}

@ARTICLE{kundu2025,
  author={Kundu, Lopamudra and Lin, Xingqin and Gadiyar, Rajesh and Lacasse, Jean-Francois and Chowdhury, Shuvo},
  journal={IEEE Communications Magazine}, 
  title={AI-RAN: Transforming RAN with AI-Driven Computing Infrastructure}, 
  year={2025},
  volume={},
  number={},
  pages={1-7},
  doi={10.1109/MCOM.001.2500018}}

@inproceedings{gao2025towards,
  title={Towards Energy-Efficient Edge Inference in Radio Cpns: a Mixture-of-Depths Transformer Based Tri-Parallel Distributed Approach},
  author={Gao, Liu and Gao, Dixiang and Xia, Nian and Peng, Mugen and Wang, Dong and Liu, Xiqing},
  booktitle={ICC 2025-IEEE International Conference on Communications},
  pages={1019--1024},
  year={2025},
  organization={IEEE}
}

@article{timor2025accelerating,
  title={Accelerating llm inference with lossless speculative decoding algorithms for heterogeneous vocabularies},
  author={Timor, Nadav and Mamou, Jonathan and Korat, Daniel and Berchansky, Moshe and Jain, Gaurav and Pereg, Oren and Wasserblat, Moshe and Harel, David},
  journal={arXiv preprint arXiv:2502.05202},
  year={2025}
}

@article{ramakrishnan2025omnidraft,
  title={OmniDraft: A Cross-vocabulary, Online Adaptive Drafter for On-device Speculative Decoding},
  author={Ramakrishnan, Ramchalam Kinattinkara and Yuan, Zhaocong and Zhuo, Shaojie and Feng, Chen and Lin, Yicheng and Su, Chenzheng and Zhang, Xiaopeng},
  journal={arXiv preprint arXiv:2507.02659},
  year={2025}
}





\end{document}